\documentclass[a4paper,11pt]{article}
\usepackage{amsmath,amsfonts,amsthm,amssymb}
\usepackage{latexsym,epsf}

 \def\picill#1by#2(#3)
 {\vbox to #2 {\hrule width #1 height 0pt depth 0pt
 \vfill\epsffile{#3}}}

\newcommand{\eq}{\begin{equation}}
\newcommand{\en}{\end{equation}}
\newcommand{\eqa}{\begin{eqnarray}}
\newcommand{\ena}{\end{eqnarray}}

\newcommand{\Z}{\mathbb{Z}}
\newcommand{\Hc}{\mathcal{H}}
\newcommand{\Pc}{\mathcal{P}}
\newcommand{\B}{\mathcal{B}}
\newcommand{\PP}{\mathcal{PB}}
\newcommand{\Ss}{\mathcal{S}}

\newcommand{\C}{\mathbb{C}}
\newcommand{\bE}{\mathbf{E}}

\numberwithin{equation}{section}
\newtheorem{thm}[equation]{Theorem}

\theoremstyle{definition}
\newtheorem{remark}[equation]{Remark}
\newtheorem{remarks}[equation]{Remarks}

\begin{document}

\setlength{\unitlength}{1mm} \thispagestyle{empty}



 \begin{center}
 {\bf \small  Abstract Error Groups Via Jones Unitary
  Braid Group Representations at q=i}

 \vspace{.2cm}

 Yong Zhang
 \footnote{yong@cs.ucf.edu} \\[.2cm]

 School of Electrical Engineering and Computer Science\\
 University of Central Florida, Orlando, FL 32816-2362
 \\[0.1cm]

 \end{center}

 \vspace{0.2cm}

\begin{center}
\parbox{14cm}{
\centerline{\small  \bf Abstract}  \noindent\\

In this paper, we classify a type of abstract groups by the central
products of dihedral groups and quaternion groups. We recognize them
as abstract error groups which are often not isomorphic to the Pauli
groups in the literature. We show the corresponding nice error bases
equivalent to the Pauli error bases modulo phase factors. The
extension of these abstract groups by the symmetric group are finite
images of the Jones unitary representations (or modulo a phase
factor) of the braid group at $q=i$ or $r=4$. We hope this work can
finally lead to new families of quantum error correction codes via
the representation theory of the braid group.

}

\end{center}

\vspace{.2cm}

\begin{tabbing}
  Unitary Braid Representation, Extraspecial Two-Groups\\[.2cm]

\end{tabbing}

 \newpage

 \section{Introduction}

Quantum error correction codes (QECC)
 are devised to
protect quantum information and computation from various kinds of
noise. These codes \cite{shor95,steane96a, crss97b, gottesman97a,
shor96}  are binary stabilizer codes with {\em abelian} normal
subgroups in the error model specified by the real or complex Pauli
groups. To explore non-binary QECC allowing arbitrary normal
subgroups later called Clifford codes, Knill
\cite{knill96a,knill96b} introduces {\em nice error bases} as well
as {\em abstract error groups}, which are refined by  Klappenecker
and R{\"o}tteler \cite{kr02a,kr02b} with the help of the Clifford
theorem \cite{gorenstein67}.

In this paper, we recognize the abstract groups exploited in
\cite{frw06,zg07,zrwwg07} as a type of the abstract error groups, as
is not done  before to the author's best knowledge. They are often
not isomorphic to the Pauli groups, though they give rise to binary
stabilized codes as the Pauli groups do. Furthermore, the extension
of these abstract error groups by the symmetric group are isomorphic
to finite images of the Jones unitary braid group representations
(UBGR) at $q=\exp(i\frac {2 \pi} r), r=4$ (or modulo a phase
factor), see \cite{jones86, frw06}. Moreover, see \cite{zg07,
zrwwg07},  some of these UBGR act as unitary basis transformation
matrices from the product basis to the Greenberger-Horne-Zeilinger
(GHZ) states \cite{ghz89}. QECC involving UBGR, for example, QECC
using GHZ states \cite{zhang08} are explored with the Shor
nine-qubit code \cite{shor95} as the simplest example.
\begin{figure}[!hbp]
\begin{center}
\epsfxsize=8.5cm \epsffile{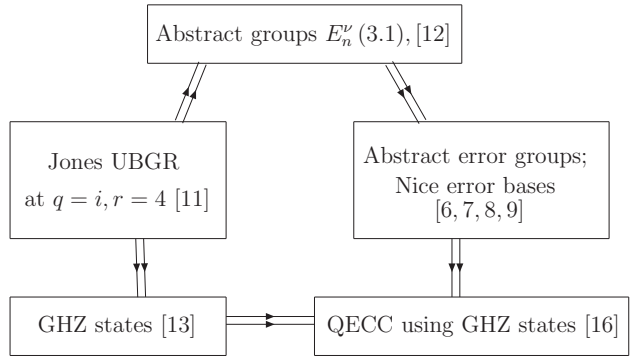} \caption{From the Jones
UBGR at $q=i, r=4$ to QECC using GHZ states} \label{fig1}
\end{center}
\end{figure}

Figure 1 is a diagrammatical exposition on the network formed by the
present paper and related references, and it also partly explains
our motivations of writing the present paper. The plan of this paper
is organized as follows. Section 2 is a preliminary on the dihedral
group, quaternion group and extraspecial two-groups. Section 3
focuses on the classification of our abstract error groups, and
compares them with the real or complex Pauli groups.  Section 4
presents associated nice error bases equivalent to the Pauli error
bases modulo phase factors, and then with them yields the Jones UBGR
at $q=i, r=4$ (or modulo a phase factor). Section 5 briefly remarks
our further research.

\section{Dihedral groups and quaternion groups}

We sketch basic facts used in the following sections about the
dihedral group $D$, quaternion group $Q$ and extraspecial two-groups
$G$. Note that relevant notations and definitions are consistent
with those in Gorenstein's book \cite{gorenstein67}.

The {\em dihedral group} $D$ with two generators $k$ and $h$ is a
set defined as \eqa
 D=\{k,h|k^2=1, h^2=-1, h k=-k h  \},
\ena whereas two generators $k$ and $h$ of the {\em quaternion
group} $Q$ satisfy very similar but distinct algebraic relations,
\eqa
 Q=\{ k,h| k^2=h^2 =-1, h k =-k h \}.
\ena Both groups have the same order $|D|=|Q|=8$ and the center
$\Z_2=\{\pm 1\}$, but they are not isomorphic to each other. The
number of cyclic subgroups of order 4 in the dihedral group $D$ is
1, whereas this number in the quaternion group $Q$ is 3. The
dihedral group $D$ is the symmetry group of a square. The quaternion
group $Q$ has a presentation in terms of quaternion, i.e.,
$\{\pm1,\pm{\bf i},\pm{\bf j},\pm{\bf k}\}$ satisfying ${\bf
i}^2={\bf j}^2={\bf k}^2={\bf ijk}=-1$.

A group $G$ of the form $G=H\circ K$ with the center $Z(G)$ is
called the {\em central product} of its two subgroups $H$ and $K$ if
and only if $hk=kh$ for all $h\in H$, $k\in K$ and $H\cap K
\subseteq Z(G)$. The central product $D^2 \equiv D \circ D$ is
isomorphic to the central product $Q^2\equiv Q \circ Q$, and hence
the central product of dihedral groups $D$ and quaternion groups $Q$
is either isomorphic to $QD^{r-1}$ or isomorphic to $D^r$, namely,
\eqa
 QD^{r-1} \equiv Q \circ \underbrace{D \circ \cdots \circ D}_{r-1},
 \quad D^{r} \equiv \underbrace{D \circ \cdots \circ D}_{r},
\ena which are not isomorphic to each other. Since $D^2\cong Q^2$,
the central product of $D$ and $Q$ is also isomorphic to  $DQ^{s-1}$
or $Q^s$, and the number of cyclic subgroups of order $4$ in
$DQ^{s-1}$ is denoted by $n$ as well as this number in $Q^s$ is by
$m$, respectively given by \eq
 \label{order_4}
 n=(2^{2s} + (-2)^s)/2, \quad m=(2^{2s} - (-2)^s)/2.
\en

{\em Extraspecial two-groups} $G$ are the central product of
quaternion groups $Q$ and dihedral groups $D$, and so there are two
extraspecial two-groups, $DQ^{r-1}$ or $Q^r$ with the same order
$2^{2 r+1}$. They can be defined in the other way if the following
hold for a group $G$: a) the center $Z(G)$ is the cyclic group
$\Z_2=\{\pm 1\}$; b) the quotient group $G/Z(G)$ is a nontrivial
elementary abelian 2-group; c) the order is $2^{2 r+1}$. With the
terminology in \cite{frw06},  a central product consisting of a
$\Z_2$ group, dihedral groups $D$ and quaternion groups $Q$ is
called {\em nearly extraspecial two-group} with the center
$\Z_2\times \Z_2$, and a central product in terms of a $\Z_4$ group,
dihedral groups $D$ and quaternion groups $Q$ is called {\em almost
extraspecial two-group} with the center $\Z_4$.

 \section{Abstract error groups $\bE^{\nu}_n$}

We identify a type of abstract groups $\bE^\nu_n$ as abstract error
groups, classify them by central products of dihedral groups,
quaternion groups, and $\Z_2$ or $\Z_4$,  and then compare them with
the Pauli groups accordingly.

 \subsection{Abstract groups $\bE^{\nu}_n$}

A finite group $G$ is called an {\em abstract error group}
\cite{kr02b} if it has an irreducible faithful unitary
representation $\phi$ with the degree defined by
$\deg\phi=|G:Z(G)|^{\frac 1 2}$, i.e., $Tr\phi(1)=|G/Z(G)|^{\frac 1
2}$. Let $d$ denote $Tr\phi(1)$. The set $\xi$ of unitary matrices,
$\xi=\{\phi(g)|\,g\in G/Z(G) \}$ forms a {\em nice error basis}
\cite{kr02a} in the $d$ dimensional Hilbert space. The quotient
group $G/Z(G)$ is called the {\em index group} associated with the
nice error basis $\xi$, with the order $|G/Z(G)|=d^2$.

 The abstract group $\bE^\nu_n$ has been recently explored
 \cite{frw06, zg07, zrwwg07}, and it is yielded by generators
 $e_1, \cdots, e_n$ satisfying:
 \eqa
 && (a) \quad  e_i^2=\nu, \quad i=1,\cdots, n, \nonumber\\
 && (b) \quad  e_i e_{i+1} =- e_{i+1} e_i, \quad i=1,\cdots, n-1, \nonumber\\
 && (c) \quad  e_i e_j =e_j e_i, \quad |i-j|\ge 2,\,\,\, i,j=1,\cdots, n
 \ena
 with $\nu$ an element of the center $Z(\bE^{\nu}_n)$, either $\nu=1$
 or $\nu=-1$. The center $Z(\bE_{2k}^\nu)$ is $\Z_2=\{\pm 1\}$ and the
 center $Z(\bE_{2k+1})$ is $\{\pm 1, \pm e_1 e_3\cdots e_{2k+1} \}$
 isomorphic to either $\Z_4$ or $\Z_2\times \Z_2$. The order of $\bE^\nu_n$ is
 $|\bE^\nu_n|=2^{n+1}$. The quotient group $\bE^\nu_n/\Z_2$ is isomorphic to
 the elementary abelian two-group, and hence $\bE^\nu_{2k}$ is
 isomorphic to extraspecial two-groups.

The abstract groups $\bE^\nu_n$ can be easily verified to be
abstract error groups, with the help of irreducible representation
theory \cite{frw06}. $\bE^\nu_{2k}$ has a $2^k$-dimensional faithful
irreducible unitary representation $\rho$ with the degree
$2^k=|\bE^\nu_{2k}/\Z_2|^{\frac 1 2}$. $\bE^{\nu}_{2k+1}$ has two
inequivalent $2^k$ dimensional faithful irreducible representations
$\lambda_1$ and $\lambda_2$ with the same degree $2^k=
|\bE^\nu_{2k+1}/Z(\bE^\nu_{2k+1})|^{\frac 1 2}$. The related nice
error bases are respectively denoted by $\rho_{2k}$,
$\lambda_{1,2k}$, $\lambda_{2,2k}$ (see Subsection 4.1).

 \subsection{The classification of $\bE^{\nu}_n$}

We firstly classify the abstract error groups $\bE^{-1}_n$ and
secondly classify $\bE^{1}_n$.

 \begin{thm}\label{cl_e_ne} Abstract error groups $\bE_n^{-1}$ are isomorphic to
 central products of quaternion groups $Q$, dihedral groups $D$, $\Z_2$
 or $\Z_4$. They are classified into three categories:
  \eqa
  && \bE^{-1}_{8j} \cong D^{4j}, \nonumber\\
  && \bE^{-1}_{8 j+1} \cong  \Z_4 \circ D^{4 j}, \nonumber\\
  && \bE^{-1}_{8 j+2} \cong  QD^{4j}, \nonumber\\
  && \bE^{-1}_{8 j+3} \cong \Z_2 \circ Q D^{4j}, \nonumber\\
  && \bE^{-1}_{8 j+4} \cong Q D^{4 j+1 }, \nonumber\\
  && \bE^{-1}_{8 j+5} \cong  \Z_4 \circ Q D^{4 j+1}, \nonumber\\
  && \bE^{-1}_{8 j+ 6} \cong  D^{4 j+3}, \nonumber\\
  && \bE^{-1}_{8 j+7} \cong \Z_2 \circ D^{4 j+3},
  \ena
  which are respectively
 isomorphic to extraspecial two-groups, almost extraspecial
 two-groups, and nearly extraspecial two-groups.
 \end{thm}

\begin{proof} We classify $\bE^{-1}_{2n}$, and then use the
 classification of $\bE^{-1}_{2n}$ to classify $\bE^{-1}_{2n+1}$. The
 proof can be read in the  three steps.

 (1). Even cases. The abstract error group $\bE^{-1}_{2n}$ is verified as
 the central product of its two subgroups $\bE^{-1}_{2n-2}$ and
 ${\bE^\prime}^{-1}_2$, i.e., $\bE^{-1}_{2n}\cong \bE^{-1}_{2n-2} \circ {\bE^\prime}^{-1}_2$.
 The group $\bE^{-1}_{2n-2}$ is generated by $e_1, e_2, \cdots, e_{2n-2}$, and
 ${\bE^\prime}^{-1}_2$ denotes the subgroup $<e_1e_3\cdots e_{2n-1}, e_{2n}>$
 generated by $e_1e_3\cdots e_{2n-1}$ and $e_{2n}$. The two generators of
 ${\bE^\prime}^{-1}_2$ satisfy the algebraic relations,
  \eqa
 && e_{2n}^2 =-1, \quad (e_1 e_3 \cdots e_{2n-1})^2=(-1)^n,
  \nonumber\\
 && (e_1 e_3 \cdots e_{2n-1}) e_{2n} =- e_{2n} (e_1 e_3 \cdots
 e_{2n-1}),
  \ena
  and hence we observe that ${\bE^\prime}^{-1}_{2}$ at $n=2k$ is isomorphic to the
  dihedral group $D$ as well as ${\bE^\prime}^{-1}_{2}$ at $n=2k+1$ is isomorphic to
  the quaternion group $Q$.

 The intersection set of $\bE^{-1}_{2n-2}$ and ${\bE^\prime}^{-1}_2$
 is $\Z_2=\{\pm 1\}$ since they have different generators.  Two generators of
 $<e_1e_3\cdots e_{2n-1}, e_{2n}>$ are commutative
 with generators of $\bE^{-1}_{2n-2}$ due to the defining
 relations of $\bE^{-1}_{2n}$. The number of group elements $g$ of $\bE^{-1}_{2n}$
 having the form $g=g_1 g_2$ with $g_1\in \bE^{-1}_{2n-2}, g_2\in {\bE^\prime}^{-1}_2$,
 is counted in the following
  \eq
     2 \times 2^{2n-2} \times 2^2 = 2^{1+2 n}
  \en
where the integers from the left to the right respectively denote
the order of the center $\Z_2$, the order of
$\bE^{-1}_{2n-2}/Z(\bE^{-1}_{2n-2})$, the order of
${\bE^\prime}^{-1}_2/Z({\bE^\prime}^{-1}_2)$, and the order of
$\bE^{-1}_{2n}$. Such a counting completes our proof that
$\bE^{-1}_{2n}=\bE^{-1}_{2n-2}\circ {\bE^\prime}^{-1}_2$, namely,
\eq
 \bE^{-1}_{4k} \cong \bE^{-1}_{4k-2}\circ D, \quad
 \bE^{-1}_{4 k +2} \cong \bE^{-1}_{4 k} \circ Q.
\en

Solving the above recursive formula between $\bE^{-1}_{2 n}$ and
$\bE^{-1}_{2n-2}$ leads to  $\bE^{-1}_{4k}\cong Q^k D^k$ and
$\bE^{-1}_{4k+2} \cong Q^{k+1} D^k$. With the isomorphic relation
$D^2 \cong Q^2$, we finally classify $\bE^{-1}_{2n}$ into four
classes,
 \eq
 \bE^{-1}_{8j} \cong D^{4j},\quad \bE^{-1}_{8j+2}\cong QD^{4j},
         \quad \bE^{-1}_{8j+4}\cong Q
 D^{4j+1}, \quad \bE^{-1}_{8j+6}\cong D^{4j+3}.
 \en

(2). Odd cases. Let us now study the classification of the abstract
error group
 $\bE^{-1}_{2n+1}$. The group element $e_1e_3\cdots e_{2n+1}$ commutes
 with all elements of $\bE^{-1}_{2n+1}$, and so it is in the center
 $Z(\bE^{-1}_{2n+1})$. The subgroup $<e_1e_3\cdots e_{2n+1}>$
 generated by $e_1e_3\cdots e_{2n+1}$ is either isomorphic to
 $\Z_4$ at $n=2k$ or isomorphic to $\Z_2$ at $n=2k+1$, due to
 $(e_1e_3\cdots e_{2n+1})^2=(-1)^{n+1}$. Hence  $\bE^{-1}_{2n+1}$ is
 the central product of its two subgroups $\bE^{-1}_{2n}$ and
 $<e_1e_3\cdots e_{2n+1}>$, namely,
 \eq
 \bE^{-1}_{4k+1} \cong \bE^{-1}_{4k}\circ \Z_4, \quad \bE^{-1}_{4k+3}
    \cong \bE^{-1}_{4k+2}\circ \Z_2
 \en
which give rise to the classification of $\bE^{-1}_{2n+1}$ as
follows
 \eqa
 && \bE^{-1}_{8j+1} \cong \bE^{-1}_{8j} \circ \Z_4, \quad \bE^{-1}_{8j+3}
     \cong \bE^{-1}_{8j+2}\circ \Z_2, \nonumber\\
 &&  \bE^{-1}_{8j+5}\cong \bE^{-1}_{8j+4}\circ \Z_4, \quad  \bE^{-1}_{8j+7} \cong
           \bE^{-1}_{8j+6} \circ \Z_2.
 \ena
Note that $\bE^{-1}_{8j+3}$ and $\bE^{-1}_{8j+7}$ have the center
$\Z_2\times \Z_2$ in terms of the first $\Z_2=\{\pm1\}$ and the
second $\Z_2=<e_1e_3\cdots e_{8j+3}>$ or $\Z_2=<e_1e_3\cdots
e_{8j+7}>$.

(3). Concluding the steps (1) and (2), we prove our theorem and
classify the abstract groups $\bE_n^{-1}$ into three categories:
$\bE^{-1}_{8j}$, $\bE^{-1}_{8j+2}$, $\bE^{-1}_{8j+4}$ and
$\bE^{-1}_{8j+6}$ are extraspecial two-groups with the center
$\Z_2$; $\bE^{-1}_{8j+1}$ and $\bE^{-1}_{8j+5}$ are almost
extraspecial two-groups with the center $\Z_4$; $\bE^{-1}_{8j+3}$
and $\bE^{-1}_{8j+7}$ are nearly extraspecial two-groups with the
center $\Z_2 \times \Z_2$.
\end{proof}

\begin{thm}
Abstract error groups $\bE^1_n$ are classified into two categories.
The ones at $n=2k$ are isomorphic to extraspecial two-groups, i.e.,
$\bE^1_{2k} \cong D^k$, and the others at $n=2k+1$ are isomorphic to
nearly extraspecial two-groups with the center $\Z_2\times \Z_2$,
i.e., $\bE^1_{2k+1}\cong \Z_2 \circ D^k$.
\end{thm}

\begin{proof}

We follow the methodology of the proof of  {\bf Theorem
\ref{cl_e_ne}}. It is easy to verify $\bE^1_2 \cong D$ and
$<e_1e_3\cdots e_{2k-1}, e_{2k}>\cong D$ which give rise to
$\bE^1_{2k}\cong D^k$. It is also explicit that $<e_1 e_3\cdots
e_{2k+1}> \cong \Z_2$ leads to $\bE^1_{2k+1} \cong \Z_2\circ D^k$
with the center $\Z_2 \times \Z_2$.
\end{proof}

\begin{remarks}
Based on the above two theorems, we are able to state that the
abstract groups $\bE^{-1}_n$ are abstract error groups which are
often not isomorphic to the Pauli groups, see the following
subsection.
\end{remarks}

\subsection{Comparisons of $\bE^\nu_n$ with the Pauli groups}

 A two-dimensional Hilbert space ${\cal H}_2 \cong \C^2$ over the
complex field ${\mathbb C}$ is called a {\em qubit} in quantum
information and computation. The symbol $1\!\! 1_2$ denotes the
$2$-dimensional identity operator or $2\times 2$ identity matrix.
The Pauli matrices $X,Y,Z$ have the conventional forms,
\eq X= \left(\begin{array}{cc} 0 & 1 \\
 1 & 0 \end{array}\right), \quad Z=\left(\begin{array}{cc} 1 & 0 \\
 0 & -1 \end{array}\right), \quad Y=ZX=\left(\begin{array}{cc}
 0 & 1 \\ -1 & 0\end{array}
 \right)
\en respectively denoting the bit-flip, phase-flip, and bit-phase
flip operations on a single qubit in quantum error correction
theory. An $n$-fold tensor product in terms of $1\!\! 1_2$ and Pauli
matrices $X,Y,Z$ has a simpler notation. For example, a 9-fold
tensor product $Z_1 Z_2 =Z\otimes Z \otimes (1\!\! 1_2)^{\otimes 7}$
where the notation $(1\!\! 1_2)^{\otimes 7}$ denotes a 7-fold tensor
product of $1\!\! 1_2$.

In the literature, the standard abstract error groups for the
 $k$-qubit Hilbert space $\Hc_2^{\otimes k}$ are the real Pauli
 group $\Pc_k$ or the complex Pauli group $\Pc_k^\prime$, and the
 related Pauli error bases are also denoted by $\Pc_k$ or
 $\Pc_k^\prime$. The Pauli group $\Pc_k$
 with the center $Z(\Pc_k)=\Z_2$ has the generators,
 \eq
 \label{pauli_basis}
  \Pc_k: \quad X_1, X_2, \cdots, X_k; Z_1, Z_2, \cdots, Z_k
 \en
which satisfy algebraic relations defining the Pauli group, \eqa
 && X_i^2 =Z_i^2 =1, \quad X_i Z_i =-Z_i X_i, \quad i=1,\cdots, k
            \nonumber\\
 && X_i X_j =X_j X_i,\quad Z_i Z_j =Z_j Z_i, \quad i,j=1,\cdots, k
           \nonumber\\
 && X_i Z_j = Z_j X_i, \quad  i \neq j  \textrm{ and } i,j=1,\cdots,
 k.
\ena Each pair of $X_i, Z_i$ yields a dihedral group $D$ so that
$\Pc_k$ is isomorphic to the central product of dihedral groups,
i.e., $\Pc_k \cong D^k$. Hence the real Pauli group $\Pc_k$ is an
extraspecial two-group.
 The complex Pauli group $\Pc_k^\prime$ is isomorphic to an almost
extraspecial two-group $\Z_4 \circ D^k$ with the center $\Z_4$
generated by the imaginary unit $i$, i.e., $\Z_4=\{\pm i, \pm 1\}$.

\begin{table}
\begin{center}
\begin{tabular}{|l|l|l|l|l|l|l|}
\hline
   & $\Pc_k$ & $\bE^1_{2k}$ & $\bE^{-1}_{2k}$ & $\Pc^\prime_k$ & $\bE^1_{2k+1}$ &
 $\bE^{-1}_{2k+1}$ \\
 \hline \hline
  $k=1$ & $D$  & $D$ & $Q$ & $\Z_4\circ D$ & $\Z_2\circ D$ & $\Z_2 \circ
  Q$\\
  \hline
\hline
  $k=2$ & $D^2$  & $D^2$ & $QD$ & $\Z_4\circ D^2$ & $\Z_2\circ D^2$ & $\Z_4
  \circ Q
  D$\\
  \hline
 \hline
  $k=3$ & $D^3$  & $D^3$ & $D^3$ & $\Z_4\circ D^3$ & $\Z_2\circ D^3$ &
  $\Z_2 \circ D^3$\\
   \hline
  \hline
 $k=4$ & $D^4$  & $D^4$ & $D^4$ & $\Z_4\circ D^4$ & $\Z_2\circ D^4$ & $\Z_4
  \circ D^4$\\
  \hline
\end{tabular}
\caption{Comparisons of $\bE^\nu_{2k}$, $\bE^\nu_{2k+1}$ with
$\Pc_k, \Pc^\prime_k$, $k=1,2,3,4$. }
\end{center}
\end{table}

In the following, we compare abstract error groups $\bE_{n}^{\nu}$
with the Pauli groups:
\begin{enumerate}
\item  The generators of the Pauli group $\Pc_k$ have the form in
terms of generators of $\bE^{-1}_{2k}$,
 \eqa
  \label{bases_trans}
 Y_i &=& e_{2 i}, \quad i =1,\cdots, k, \nonumber\\
 Z_i &=& (-\sqrt{-1})^i e_1 e_3 \cdots e_{2i-1}, \nonumber\\
 X_i &=& Z_i Y_i  =(-\sqrt{-1})^i e_1 e_3 \cdots e_{2i-1} e_{2i};
 \ena
\item The real Pauli groups $\Pc_k$ are isomorphic to $\bE^1_{2k}$;
\item The complex Pauli groups $\Pc^\prime_k$ with the center $\Z_4$
are not isomorphic to $\bE^1_{2k+1}$ with the center $\Z_2\times
\Z_2$; \item  $\bE^{-1}_{8j}, \bE^{-1}_{8j+6} $ are respectively
isomorphic to $\Pc_{4j}, \Pc_{4j+3}$; \item $\bE^{-1}_{8j+2},
\bE^{-1}_{8j+4} $ are very interesting cases since they contain a
quaternion group $Q$ which is not in the Pauli groups; \item
$\bE^{-1}_{8j+1}$ are isomorphic to the complex Pauli groups
$\Pc^\prime_{4j}$; \item $\bE^{-1}_{8j+7}$ are isomorphic to
$\bE^{1}_{8j+7}$, and both have the same center $\Z_2\times \Z_2$.

\end{enumerate}

As an example, we compare $\Pc_2$ with $\bE^{-1}_4$. With the help
of the formula (\ref{order_4}), we observe that the real Pauli group
$\Pc_2\cong D^2$ has 20 order-2 elements,
 \eqa
 &&  \pm 1, \quad \pm X_1, \quad \pm X_2,\quad  \pm Z_1,\quad \pm Z_2,
 \nonumber\\
 && \pm X_1 X_2, \quad \pm X_1 Z_2, \quad \pm X_2 Z_1, \quad \pm Z_1
 Z_2,\quad
 \pm X_1 X_2 Z_1 Z_2
 \ena
 and $12=2^4-(-2)^2$  order-4 elements,
 \eq
\pm X_1 Z_1, \quad \pm X_2 Z_2, \quad \pm X_1 X_2 Z_1, \quad \pm X_1
Z_1 Z_2,\quad \pm X_1 X_2 Z_2, \quad \pm X_2 Z_1 Z_2.
 \en
whereas the abstract error group $\bE^{-1}_4\cong DQ$ has 12 order-2
elements,
 \eqa
\pm 1, \quad \pm e_1 e_3, \quad \pm e_1 e_4, \quad \pm e_2 e_4,\quad
\pm e_1 e_2 e_4, \quad \pm e_1 e_3 e_4,
 \ena
and $20=2^4+(-2)^2$  order-4 elements, \eqa
 && \pm e_1, \quad \pm e_2, \quad \pm e_3, \quad \pm e_4, \quad \pm
 e_1 e_2,\quad \pm e_2 e_3, \quad \pm e_3 e_4, \nonumber\\
 && \pm  e_1 e_2 e_3, \quad \pm e_2 e_3 e_4, \quad \pm e_1 e_2 e_3
 e_4.
\ena

Table 1 lists more examples for comparisons of our abstract error
groups $\bE^\nu_{2k}$, $\bE^\nu_{2k+1}$ with the Pauli groups
$\Pc_k, \Pc_k^\prime$, $k=1,2,3,4.$

\begin{remark}
The index groups $\bE^\nu_n/\Z_2$ are elementary abelian groups and
so associated QECC are still stabilizer codes according to the
Clifford code theory \cite{kr02a}, which is the same as the Pauli
groups.
\end{remark}

\section{Nice error bases and Jones UBGR at $q=i, r=4$}

We recognize unitary irreducible representations of abstract error
groups $\bE^\nu_n$ as nice error bases, and set up a Jones UBGR at
$q=i, r=4$ (or modulo a phase factor) in terms of these bases.

\subsection{Nice error bases associated with $\bE^{\nu}_n$ }

Given an abstract error group $G$ and its faithful irreducible
representation $\phi$, the associated nice error basis
$\xi=\{\phi(g)|\, g\in G\}$ satisfies the defining relations
\cite{kr02a}:
\begin{description}
\item a). $\phi(1)$ is the $d\otimes d$ identity matrix, $d=Tr\phi(1)$;
\item b). $Tr\, \phi(g) =0$, for all $g\in
G/Z(G)$ and $g \neq 1$;
\item c). $\phi(g)\, \phi(h) =\omega(g,h)\, \phi(gh)$ for all $g,h\in
G/Z(G)$;
\end{description}
where the set of complex numbers $\omega(g,h)$ forms a cyclic group,
and hence the nice error basis $\xi$ is a projective faithful
irreducible representation of the index group $G/Z(G)$.

Now we study nice error bases associated with the abstract error
groups $\bE^{\nu}_n$. Interested readers are invited to see
\cite{frw06} for irreducible representations of $\bE^{\nu}_{n}$. 1)
Even cases at $\nu=-1$. The $2^k$-dimensional irreducible
representation $\rho$ of $\bE^{-1}_{2k}$ has the form,
 \eqa \label{rho} \quad\quad \rho(e_1) &=& \sqrt{-1} Z_1, \nonumber\\
            &\vdots&  \nonumber\\
            \rho(e_{2i-1}) &=& \sqrt{-1} Z_{i-1} Z_i,\quad i=1,2,\cdots, k,
                \nonumber\\
            \rho(e_{2i}) &=& Y_i, \nonumber\\
            &\vdots& \nonumber\\
            \rho(e_{2k}) &=& Y_k.
  \ena
2) Odd cases at $\nu=-1$. The two inequivalent $2^k$-dimensional
irreducible representations of $\bE^{-1}_{2k+1}$ are respectively
denoted by $\lambda_i$, $i=1,2$, in which $\lambda_i(e_1)$,
$\lambda_i(e_2)$,$\cdots$, $\lambda_i(e_{2k})$ have the same form as
in the representation $\rho$, except that $\lambda_i(e_{2k+1})$ is
specified to be $\lambda_i(e_{2k+1})= \pm \sqrt{-1} Z_k$, the $+$
sign for $\lambda_1$ and $-$ sign for $\lambda_2$.

The irreducible representation $\rho$ of $\bE^{-1}_{2k}$ leads to
the nice error bases  $\rho_{2k}$ given by
 $$
 \rho_{2k}=\{\rho(g)| g\in \bE^{-1}_{2k}/Z(\bE^{-1}_{2k})\}
 $$
which is found to satisfy the above defining relations of nice error
bases. The nice error bases associated with irreducible
representations $\lambda_1, \lambda_2$ of the abstract error group
$\bE^{-1}_{2k+1}$ are respectively denoted by $\lambda_{1,2k},
\lambda_{2,2k}$. It is explicit that $\lambda_{1,2k},
\lambda_{2,2k}$ represent the same nice error bases as $\rho_{2k}$
in the $2^k$-dimensional vector space.

As all generators of $\bE^{-1}_n$ are rescaled by the imaginary unit
$i$, the resulted group is isomorphic to the abstract error group
$\bE^1_n$. Hence nice error bases $\rho^\prime_{2k}$ at $n=2k$ and
$\lambda^\prime_{i,2k}$, $i=1,2$, at $n=2k+1$ associated  with
$\bE^1_n$ are respectively obtained by $\rho_{2k}$, $\lambda_{1,2k}$
and $\lambda_{2,2k}$ times the imaginary unit $i$.

\begin{remark}
Nice error bases $\rho_{2k}$ (\ref{rho}) are recast in terms of the
Pauli matrices with relevant phase factors, though associated
abstract error groups $\bE^{-1}_{2k}$ are often not isomorphic to
the Pauli groups $\Pc_k$. $\rho_{2k}$ (\ref{rho}) can be proved to
be equivalent to the Pauli error bases $\Pc_k$ (\ref{pauli_basis})
modulo phase factors, in view of the definition \cite{kr03} for the
equivalence of two unitary error bases.
\end{remark}

\subsection{Jones UBGR at $q=i, r=4$ via nice error bases}

Artin's braid group $\B_n$ on $n$ strands has a presentation in
terms of generators $b_1,\ldots$,$b_{n-1}$ satisfying the
commutation relation, \eq \label{br1} b_{i} b_{j} = b_{j}
b_{i},\qquad |i - j | \ge 2 \en and the braid relations,
 \eq \label{br2}
   b_{i}  b_{i+1} b_{i} = b_{i+1} b_{i} b_{i+1},
   \qquad 1 \leq i \leq n-2.
\en

The symmetric group $\Ss_n$ includes all possible permutations of
$n$ objects, and it is generated by transpositions $(i,i+1)$ between
$i$-th object and $(i+1)$-th object, $1\le i \le n-1$. $\B_n$ has a
finite-index normal subgroup $\PP_n$ generated by all conjugates of
the squares of the generators of $\B_n$, see \cite{frw06}. $\PP_n$
is called the \emph{pure braid group} and can be understood as the
kernel of the surjective homomorphism $\B_n\to\Ss_n$ given by
$b_i\mapsto (i,i+1)$. In other words we have an isomorphism
${\mathcal S}_n$ $\cong$ ${\mathcal B}_n/\PP_n$.

\begin{thm}\label{zrwwg}
 Let $\{T_1,\ldots,T_{n-1}\}$ be the images
of the generators $\{e_1,\ldots,e_{n-1}\}$ of $\bE^{-1}_{n-1}$ under
a representation $\phi_{n-1}$ of $\bE^{-1}_{n-1}$ such that:
\begin{enumerate}
\item[(a)] $T_i^2=-Id$,
\item[(b)] $T_iT_j=T_jT_i$ if $|i-j|>1$,
\item[(c)] $T_iT_{i+1}=-T_{i+1}T_i$ for all $1\leq i\leq (n-2)$.
\end{enumerate}
Then the set of matrices $\{\check{R}_1,\ldots,\check{R}_{n-1}\}$
defined by $\check{R}_i=\frac{1}{\sqrt{2}}(Id+T_i)$
 gives a representation of $\B_n$ by $b_i\rightarrow \check{R}_i$.
 If in addition
the $T_i$ are anti-Hermitian (i.e. $T_i=-T_i^\dag$), the $\B_n$
representation is unitary.
\end{thm}

\begin{proof}
This theorem is firstly used \cite{zg07} and then proved as {\bf
Lemma 3.8} in \cite{zrwwg07}.
\end{proof}

In terms of the $2^k$-dimensional nice error bases $\rho_{2k}$,
namely, $T_i=\rho(e_i)$, the finite image $\hat{\rho}_{2k}$
\cite{frw06} of the irreducible UBGR has the form, \eqa
\label{r_matrices}
 \hat{\rho}_{2k}:\quad\quad \check{R}_1 &=& d_1, \nonumber\\
            &\vdots&  \nonumber\\
            \check{R}_{2i-1} &=& D_{i-1,i} ,\quad i=2,\cdots, k, \nonumber\\
            \check{R}_{2i} &=& f_i, \nonumber\\
            &\vdots& \nonumber\\
           \check{R}_{2k} &=& f_k,
  \ena
where $d=e^{i\frac \pi 4 Z}$, $f= e^{\frac \pi 4 Y}$ and lower
indices of $d_1$, $f_i$, $f_k$ have the same meaning as that of
$Z_1$, $Z_i$, $Z_k$, and the $k$-fold tensor product $D_{i-1,i}$
 denotes a form
 $$ D_{i-1,i} = (1\!\! 1_2)^{\otimes {i-1}} \otimes  D
  \otimes (1\!\! 1_2)^{\otimes k-i-1}, \quad  D=e^{i\frac \pi 4 Z\otimes Z}.$$
Note that the finite images $\hat{\lambda}_{1,2k}$ or
$\hat{\lambda}_{2,2k}$ are associated with nice error bases
$\lambda_{1,2k}$ and $\lambda_{2,2k}$, and  $\check{R}_{2k+1}=d_k$
for $\hat{\lambda}_{1,2k}$ or $\check{R}_{2k+1}=\bar{d_k}$ for
$\hat{\lambda}_{2,2k}$ where $\bar{}$ denotes the complex
conjugation.

The irreducible representation $(\pi_{2k+1}, (\C^2)^{\otimes {2k+1}}
)$ of the braid group $\B_{2k+1}$ under the  $\check{R}$ matrices
(\ref{r_matrices}) are defined by $\pi_{2k+1}(b_i)=\check{R}_i$. The
finite image of the pure braid group $\PP_n$ under this
representation is denoted by $H_{2k+1}=\pi_{2k+1}(\PP_{2k+1})$ which
is isomorphic to the abstract error group $\bE^{-1}_{2k}$, namely,
$H_{2k+1}\cong \bE^{-1}_{2k}$. The finite image of the braid group
$\B_{2k+1}$ under this representation denoted by
$G_{2k+1}=\pi_{2k+1} (\B_{2k+1})$ is the extension of
$\bE^{-1}_{2k}$ by the symmetric group $\Ss_{2k+1}$, namely,
$G_{2k+1}/\bE^{-1}_{2k} \cong \Ss_{2k+1}$. Furthermore, the Jones
UBGR at $q=i, r=4$ \cite{jones86} is obtained by rescaling the
$\check{R}$ matrices (\ref{r_matrices}) in the way
$\check{R}^\prime_i = -e^{-i\frac \pi 4} \check{R}_i$. The images of
the braid group $\B_{2k+1}$ and the pure braid group $\PP_{2k+1}$
under the representation $\pi^\prime_{2k+1}(b_i)=\check{R}^\prime_i$
are respectively denoted by $$G_{2k+1}^\prime
=\pi^\prime_{2k+1}(\B_{2k+1}), \quad H_{2k+1}^\prime
=\pi^\prime_{2k+1}(\PP_{2k+1}) $$ It is easy to verify
$H_{2k+1}^\prime \cong \bE^1_{2k}$ and $G_{2k+1}^\prime/\bE^1_{2k}
\cong \Ss_{2k+1}$. Interested readers are invited to consult
\cite{jones86, frw06} for details.

\begin{remark}
The Jones UBGR at $q=i, r=4$ (or modulo a phase factor) can be
regarded as unitary basis transformation matrices from the product
states to GHZ states \cite{ghz89}, see Ref. \cite{zg07, zrwwg07} and
the references therein. Furthermore, the author has already
initiated the project of constructing QECC involving unitary braid
representations, for examples, QECC using GHZ states \cite{zhang08}.
\end{remark}

 \section{Concluding remarks}

 In this paper, we recognize finite images of the Jones unitary pure
 braid group representations at $q=i, r=4$ as a type of abstract error groups.
 We classify them by central products of dihedral groups and quaternion
 groups, and then realize that they are often not isomorphic to the Pauli
 groups in the literature. In our further research, we will continue to
 explore interesting QECC associated with unitary representations of the
 braid group and develop related fault-tolerant quantum computation, for
 examples, QECC using GHZ states \cite{zhang08}. Furthermore, we will revisit
 topics using the Pauli error bases in quantum information and computation
 with our nice error bases (\ref{rho}), for examples, quantum teleportation
 \cite{werner01} and noiseless subsystems \cite{klv00}.

 \section*{Acknowledgements}

 Y. Zhang  thanks Jennifer Franko, Eric C. Rowell, Zhenghan Wang
 and Yong-Shi Wu for stimulating comments, and thanks
 Pawel Wocjan for both helpful comments on the manuscript and
 financial support with his NSF grants CCF-0726771 and CCF-0746600,
 and thanks anonymous referees for their helpful comments.
 The author is also in part supported by
 NSF-China Grant-10605035.

\end{document}